\documentclass[a4paper,12pt]{article}
\usepackage{latexsym,amssymb,amsfonts,amsmath,amsthm}
\usepackage[dvips]{graphicx}
\usepackage{color}
\usepackage{ulem}

\renewcommand{\theequation}{\thesection.\arabic{equation}}

\newtheorem{theorem}{Theorem}
\newtheorem{lemma}{Lemma}

\newtheorem{proposition}{Proposition}
\newtheorem{remark}{Remark}
\newtheorem{definition}{Definition}

\newtheorem{assumption}{Assumption}

\numberwithin{theorem}{section}
\numberwithin{lemma}{section}
\numberwithin{corollary}{section}
\numberwithin{proposition}{section}
\numberwithin{remark}{section}

\newcommand{\bs}[1]{\boldsymbol{#1}}

\title{
Eigenvalues of quantum walk induced by recurrence properties of the underlying birth and death process: application to computation of an edge state
}
\author
{
{Yusuke IDE,$^{1}$\footnote{E-mail: ide@neptune.kanazawa-it.ac.jp} \quad Norio KONNO,$^{2}$\footnote{E-mail: konno@ynu.ac.jp} \quad Etsuo SEGAWA$^{3}$\footnote{E-mail: segawa-etsuo-tb@ynu.ac.jp}}\\
\noindent \\
{\scriptsize {}$^{1}$ \textit{Math. and Science Education Research Center, Kanazawa Institute of Technology, }}\\  
{\scriptsize \textit{Nonoichi 921-8501, Japan}}\\
\noindent \\
{\scriptsize {}$^{2}$ \textit{Department of Applied Mathematics, Faculty of Engineering,}}\\  
{\scriptsize {}$^{3}$ \textit{Graduate School of Education Center}, {}$^{3}$ \textit{Graduate School of Environment and Information Sciences, }}\\ 
{\scriptsize \textit{Yokohama National University,}}\\ 
{\scriptsize \textit{Yokohama 240-8501, Japan}}\\
}
\date{}
\begin{document}
\maketitle
%%%%%%%%%%%%%%%%%%%%%%%%%%%%%%%%%%%%%%%%%%%%%%%%%%%%%%%%%%%%%%%%%%%%%%%%
%%Abst%%%%%%%%%%%%%%%%%%%%%%%%%%%%%%%%%%%%%%%%%%%%%%%%%%%%%%%%%%%%%%%%%%
\noindent
\begin{small}
\textbf{Abstract. }
In this paper, we consider an extended coined Szegedy model and discuss the existence of the point spectrum of induced quantum walks in terms of recurrence properties of the underlying birth and death process. 
We obtain that if the underlying random walk is not null recurrent, then the point spectrum exists in the induced quantum walks. 
As an application, we provide a simple computational way of the dispersion relation of the edge state part for the topological phase model driven by quantum walk 
using the recurrence properties of underlying birth and death process. 
\end{small}
%%%%%%%%%%%%%%%%%%%%%%%%%%%%%%%%%%%%%%%%%%%%%%%%%%%%%%%%%%%%%%%%%%%%%%%%
%%%%%%%%%%%%%%%%%%%%%%%%%%%%%%%%%%%%%%%%%%%%%%%%%%%%%%%%%%%%%%%%%%%%%%%%
\section{Introduction}
The Szegedy model~\cite{Sze} is one of the intensively investigated quantum walk models because 
the eigenproblem 
is reduced to that of the underlying random walk. 
Using this fact of the spectrum, we can see the effectiveness of quantum search driven by this quantum walks on some finite graphs~e.g., \cite{PorBook} and references therein, 
and also limit distribution on some infinite graphs~e.g., \cite{HKSS}. 
The spectral map from the induced quantum walk to the underlying random walk is obtained by the Joukowsky transform~\cite{HKSS}. 
In that sense, quantum walk contains all the information of the underlying random walk. 
Then the following natural question arises: 
\begin{center} {\it How does the underlying random walk's behavior affect the induced quantum walk ?}\end{center}
As a trial to answer this question, in this paper, we treat recurrence properties of the underlying random walk on the half line. 
The recurrence property of random walk can be classified into the three classes: positive recurrent; null recurrent; and transient~\cite{Schi}
(see Definition~\ref{def:RW} for more detail). 
In this paper, we treat a coined walk which is converted from the Szegedy model
\footnote{In \cite{KPSS}, it is shown that coined quantum walks are unitarily equivalent to several kinds of quantum walk models 
such as $2$-staggered walk~\cite{Por}, the Szegedy model~\cite{Sze}, split step quantum walks and so on 
by some graph deformations of the original graph~\cite{KPSS}. }.
Here a coined walk is defined by pair of a connected graph 
$G=(V, E)$ and a sequence of unitary matrices $(C_u)_{u\in V}$ assigned to each vertex. 
The detailed definition of this coined walk, named extended coined-Szegedy model and the underlying random walk is denoted by Section~\ref{Sec:2}. 
Our obtained main result is as follows. 
\begin{theorem}\label{thm:main}
For the extended coined-Szegedy model on the half line, 
if the underlying random walk is not null recurrent, then the induced quantum walk has point spectrum. 
\end{theorem}
Therefore if the underlying random walk does not have continuous singular spectrum and not null recurrent, then localization happens~\cite{SuzSeg}. 
Moreover if the underlying random walk is positive recurrent, then the point spectrum of the quantum walk 
is derived from the stationary distribution of this random walk (see Proposition~\ref{prop:1}). 
On the other hand, if the underlying random walk is transient, then the point spectrum of the quantum walk 
is derived from the finite energy flow~\cite{HS} which cannot be expressed by the eigenspace inherited by the random walk (see Proposition~\ref{prop:2}). 
As an application of this theorem, we compute the dispersion relation of the topological phase simulator model 
driven by an alternative quantum walk~\cite{AE} which is determined by the parameters $\alpha,\beta\in [0,2\pi)$. 
We show that the edge state of this model is described by the recurrence properties of the underlying random walk (see Theorem~\ref{thm:dispersion} and also Fig.~\ref{fig:dispersion}). 
Since the recurrence property of the random walk is determined by the global structure (see Theorem~\ref{thm:RW}), 
then the edge state is stable to small spatial perturbation of parameters $\alpha,\beta$ with respect to the horizontal direction. 

This paper is organized as follows. 
In Section~2, the definition of our quantum walk model is explained and also we introduce an extended detailed balanced condition to connect the underlying random walk. 
In Section~3, we restrict the treated graph to the half line. 
We show the existence of the point spectrum of the induced quantum walk in terms of the underlying random walk's behavior 
which is corresponding to the proof of Theorem~\ref{thm:main}. 
In Section~4, the result is applied to the topological phase model, and we obtain another aspect of the edge state from the random walk's point of view. 
%%%%%%%%%%%%%%%%%%%%%%%%%%%%%%%%%%%%%%%%%%%%%%%%%%%%%%%%%%%%%%%%%%%%%%%%
%%%%%%%%%%%%%%%%%%%%%%%%%%%%%%%%%%%%%%%%%%%%%%%%%%%%%%%%%%%%%%%%%%%%%%%%
\section{Extended coined-Szegedy model}\label{Sec:2}
Let $G=(V,E)$ be a connected graph which may has infinite number of vertices but whose degree is uniformly bounded. 
The set of symmetric arcs induced by the edge set $E$ is denoted by $A$.
Here if $\{u,v\}\in E$, then the induced symmetric arcs are $(u,v)$ and $(v,u)$ which represent the arcs from the origin vertex $u$ to 
the terminal vertex $v$ and vice versa, respectively. 
The inverse arc of $a\in A$ is $\bar{a}$ and the origin and terminus vertices of $a$ are $o(a)$ and $t(a)$, respectively. 
Remark that $t(\bar{a})=o(a)$ and $o(\bar{a})=t(a)$ hold.
The support edge of $a$ is denoted by $|a|\in E$; that is, for $a=(u,v)\in A$, $|a|=|\bar{a}|=\{u,v\}$. 

The total Hilbert space of quantum walk is 
$\mathcal{A}:=\ell^2(A)$. The inner product is standard. 
We take the subspace $\mathcal{A}_u=\{\psi\in \mathcal{A} \;|\; t(a)\neq u \Rightarrow \psi(a)=0 \}$, and it holds 
$\mathcal{A}=\oplus_{u\in V}\mathcal{A}_u$. 
%We define $\chi_u: \mathcal{A}_u\to \mathcal{A}$ $(u\in V)$ by 
%	\[ (\chi_u\psi)(a)=\begin{cases} \psi(a) & \text{: $t(a)=u$}\\ 0 & \text{: otherwise.} \end{cases}\]
%The adjoint $\chi^*_u: \mathcal{A}\to \mathcal{A}_u$ is described by $(\chi_u\psi)(a)=\psi(a)$ for any $a\in A$ with $t(a)=u$. 
The local unitary operator on $\mathcal{A}_u$ is denoted by $C_u$ and we describe the coin operator $C=\oplus_{u\in V}C_u$ 
under the decomposition of $\mathcal{A}=\oplus_{u\in V}\mathcal{A}_u$. 

Let $S$ be the shift operator such that $(S\psi)(a)=\psi(\bar{a})$. 
Then the total time evolution of our quantum walk is $U=SC$. 
To extract some relation between quantum walk and random walk behaviors,  
we take the following assumption which is motivated by \cite{HIKST}: 
\begin{assumption}\label{Ass}
Let $\sigma\in (0,2\pi)$ be a real valued constant. 
\begin{enumerate}
\item The eigenvalues of $C_u$ are given by $\mathrm{Spec}(C_u)=\{1,e^{i\sigma}\}$ for any $u\in V$;  
\item $\dim\ker(1-C_u)=1$
\end{enumerate}
\end{assumption}
\begin{remark}\label{rem:rotate}
\noindent
\begin{enumerate}
\item If $\sigma=\pi$, then the coined-Szegedy model appears~\cite{Seg}. 
\item We call Assumption 1' by just replacing Assumption~\ref{Ass} (1) into $\sigma(C_u)=\{e^{i\sigma_1},e^{i\sigma_2}\}$. 
This model can be represented by the $\sigma=\sigma_2-\sigma_1$ case with the 
modification of $U'=e^{\sigma_2}U$ for the Assumption~1 model. 
\end{enumerate}
\end{remark}
For $u\in V$, let $\alpha_u\in \ell^2(\mathcal{A}_u)$ with $||\alpha_u||=1$ be normalized $(+1)$-eigenvector of $C_u$; that is, 
$\ker(1-C_u)=\mathbb{C}\{\alpha_u\}$. 
Using $\alpha_u$, we define a weight $\alpha: A\to \mathbb{C}$ on each symmetric arc by $\sum_{t(a)=u}|\alpha(a)|^2=1$ for any $u\in V$
such that $\alpha(a)=\alpha_{t(a)}(a)$. 
Let us $K:\ell^2(V)\to\ell^2(A)$ such that $(Kf)(a)=\alpha(a)f(t(a))$. The adjoint $K^*$ is described by $(K^*\psi)(u)=\langle \alpha_u, \psi \rangle$. 
Note that $K$ is an isometric operator, that is, $||Kf||^2=||f||^2$. 
Then we have the following lemma. 
\begin{lemma}
Let $\mathcal{L}\subset \mathcal{A}$ be denoted by $K\ell^2(V)+SK\ell^2(V)$. Then we have 
\begin{enumerate}
\item $U(\mathcal{L}) = \mathcal{L}$; 
\item $\mathcal{L}^\perp = \left(\ker(1-S)\cap \ker K^*\right) \oplus \left(\ker(-1-S)\cap \ker K^*\right)$ holds. Under this decomposition, we have 
$U|_{\mathcal{L}^\perp}=e^{i\sigma}\oplus (-e^{i\sigma})$. 
\end{enumerate}
\end{lemma}
\begin{proof}
Remark that $K$ is an isometric operator, that is, $K^*K=1_V$. 
Using this fact, we obtain 
	\[ UL=L\tilde{T}, \] 
where $L:\ell^2(V)\oplus \ell^2(V)\to \ell^2(A)$ such that $L[f \;\; g]^\top=Kf+SKg$ for $f,g\in \ell^2 (V)$ and 
	\begin{equation}\label{eq:tildeT} 
        \tilde{T}_\sigma=\begin{bmatrix} 0 & e^{i\sigma} \\ 1 & (1-e^{i\sigma})T \end{bmatrix}. 
        \end{equation}
Here $T=K^*SK$. 
Since $\tilde{T}_\sigma$ has the inverse, then $\tilde{T}_\sigma:\ell^2(V)\oplus \ell^2(V)\to \ell^2(V)\oplus \ell^2(V)$ is bijection. 
Therefore the proof of part 1 is completed. 
For the part 2 of the proof, let us consider $\mathcal{L}^\perp$. 
It holds
\[ \mathcal{L}^\perp= [(\ker(K^*)\cap \ker(K^*S)) \cap \ker(1-S)] \oplus [(\ker(K^*)\cap \ker(SK^*)) \cap \ker(1+S)]. \]
Since $K^*S\psi=\pm K^*\psi=0$ for any $\psi \in \ker(K^*) \cap \ker(\pm 1-S)$, we have 
\[ (\ker(K^*)\cap \ker(K^*S)) \cap \ker(\pm 1-S)=\ker(K^*) \cap \ker(\pm 1-S). \]
For any $\varphi\in \ker(K^*) \cap \ker(\pm 1-S)$, it is easy to check that $U\psi=\pm e^{i\sigma} \psi$ since 
$U=S(KK^*+e^{i\sigma}(1-KK^*))$. 
Then it is completed the proof of part 2. 
\end{proof}
We put $T:=K^*SK$ which will play an important role. It is easy to see that 
	\[ (Tf)(u)=\sum_{a:t(a)=u}\alpha(\bar{a})\overline{\alpha(a)}f(o(a)). \]
\begin{definition}
If there is a function $m_V:V\to \mathbb{C}\setminus\{0\}$ such that 
\begin{equation}\label{eq:DBC} 
\alpha(\bar{a})m_V(o(a))=\alpha(a)m_V(t(a))\;({}^\forall a\in A),  
\end{equation}
then we say $T$ is reversible and $m_V$ is called a reversible measure.
\end{definition}
Let $P$ be the probability transition operator defined by 
	\[(Pf)(u)=\sum_{o(a)=u}p(a)f(t(a)),\] where $p(a):=|\alpha(\bar{a})|^2$. 
Here the matrix expression is $(P)_{u,v}=\sum_{o(a)=u,t(a)=v}|\alpha(\bar{a})|^2$, which describes the transition 
probability from $u$ to $v$. 
If $T$ is reversible, since $\alpha(\bar{a})\overline{\alpha(a)}=|\alpha(a)|^2 m_V(t(a))/m_V(o(a))$, then $T$ is reexpressed by 
$T=D^{-1}PD$, where $(Df)(u)=1/m_V(u)$.  

The $(\pm e^{i\sigma})$-eigenspace of $\mathcal{L}^\perp$ are denoted by $\mathcal{H}_\pm$, respectively. 
Then $\psi\in \mathcal{H}_\pm$ iff 
	\begin{align} 
        \psi(a)\pm \psi(\bar{a}) &= 0\;({}^\forall a\in A), \;\mathrm{(Skew\;symmetricity(+), Symmetricity(-), respectively)} \label{symmetry}\\
        \langle \alpha_u,\psi \rangle &= 0\;({}^\forall u\in V), \label{Kirch}
        \end{align}
respectively. 
If $T$ is reversible, we define $m_E:E\to \mathbb{C}$ by 
$m_E(e):=\alpha(\bar{a})m_V(o(a))=\alpha(a)m_V(t(a))$, where $|a|=e$. 
Then the condition (\ref{Kirch}) is reduced to 
	\begin{align}
        \sum_{t(a)=u}\psi(a)m_E(|a|) &= \sum_{o(a)=u}\psi(a)m_E(|a|)=0, \;(\mathrm{Kirchhoff\;condition})
        \end{align}
where we used (\ref{symmetry}) in the second equality. 

\begin{lemma}\label{lem:PT}
If $T$ is reversible, then $P$ is also reversible, and if the reversible measure of $T$ is $\ell^2(V)$, then $\ker(1-T)=\mathbb{C}\{m_V\}$. 
\end{lemma}
\begin{proof}
Just taking the square modulus to both sides of (\ref{eq:DBC}), we obtain the former. 
It is easy to check that $Tm_V=m_V$ holds if $T$ is reversible. 
Since $T$ is isomorphic to the transition operator $P$ if $T$ is reversible, the Perron-Frobeniou theorem implies $\ker(1-T)=\mathbb{C}\{m_V\}$
which completes the proof of the latter. 
\end{proof}
The following lemma is obtained by tracing the proofs of spectral mapping theorems of quantum walks in \cite{HPSS, MOS} 
with some modifications of the settings. 
%%%
\begin{lemma}\label{lem:recQW}
\noindent
$T$ is reversible and $m_V\in \ell^2(V)$ if and only if $U|_\mathcal{L}$ has the eigenvalue $1$. 
In this case, $\ker(1-U|_\mathcal{L})=\mathbb{C}\{Km_V\}$. 
\end{lemma}
%%%
\begin{proof}
Assume that $T$ is reversible and $m_V\in \ell^2(V)$. 
Remark that the coin operator $C$ is expressed by $C=KK^*+e^{i\sigma}(1-KK^*)$ and $K^*K=1_V$. Then $UKm_V=SKm_V$ holds. 
Since $m_V$ is a reversible measure, it holds that 
\begin{equation}\label{eq:ibuki} (SKm_V)(a)=(Km_V)(\bar{a})=\alpha(\bar{a})m_V(o(a))=\alpha(a)m_V(t(a))=(Km_V)(a) \end{equation}
for any $a\in A$. 
Then we have $UKm_V=SKm_V=Km_V$. Since $K$ is an isometric operator and $m_V\in \ell^2(V)$, we have $Km_V\in \ell^2(A)$, 
which implies that $U|_\mathcal{L}$ has the eigenvalue $1$. 

On the other hand, let us assume that $U|_\mathcal{L}$ has an eigenvalue $1$ and 
take arbitrarily $\psi\in \ker(1-U|_\mathcal{L})$. 
We will show that $T$ is reversible, $m_V\in\ell^2(V)$ and $\psi\in \mathbb{C}\{Km_V\}$ as follows. 
Since $\psi\in \mathcal{L}$, there exists $\varphi:=[\phi_1,\phi_2]^\top\in \ell^2(V)\oplus \ell^2(V)$ such that $\psi=K\phi_1+SK\phi_2$, that is, 
$\psi=L\varphi$. 
By (3.18) in \cite{MOS}, $UL\varphi=L\varphi$ includes 
	\begin{equation}\label{eq:condition}
         L(1-\tilde{T}_\sigma)\varphi = 0,\;\varphi \notin \ker(L) 
        \end{equation}
and 
	\begin{equation}\label{eq:condition2} 
        \ker L=\ker(1-\tilde{T}_\pi^2)=\ker \begin{bmatrix} 1 & T \\ T & 1 \end{bmatrix}
        =\left\{ \begin{bmatrix} f+g \\ -f+g \end{bmatrix} \;\bigg|\; f\in\ker(1-T),\;g\in\ker(1+T)\right\}. 
        \end{equation}
Thus the conditions (\ref{eq:condition}) and (\ref{eq:condition2}) imply  
	 \[ \begin{bmatrix} 1 & T \\ T & 1 \end{bmatrix} (1-\tilde{T}_\sigma)\varphi=0,\;\varphi\notin \ker \begin{bmatrix} 1 & T\\ T & 1\end{bmatrix}, \]
which can be equivalently transformed by the Gaussian elimination as follows: 
         \begin{equation}\label{eq:ibuki2} 
         \begin{bmatrix} 1-T & -(1-T) \\ 0 & 1-T^2 \end{bmatrix}\varphi =0\;\;\mathrm{and}\;\;\begin{bmatrix} 1 & T \\ T & 1 \end{bmatrix}\varphi\neq 0.   
         \end{equation}
In the following, let us see this condition is more concretely described by 
	\begin{equation}\label{eq:inverse} 
        \varphi\in \left\{ \begin{bmatrix} f_1+g \\ f_2+g \end{bmatrix} \;\bigg|\; f_1\neq -f_2,\;f_1,f_2\in \ker(1-T),\;g\in\ker(1+T) \right\}. 
        \end{equation}
Let $\varphi=[\phi_1\;\phi_2]^\top$. 
The first condition in (\ref{eq:ibuki2}) can be rewritten by  
\begin{enumerate}
\item $\phi_2=f_2+g_2$, where $f_2\in \ker(1-T)$, $g_2\in \ker(1+T)$, and 
\item $(1-T)\phi_1=(1-T)\phi_2$.  
\end{enumerate}
Since $(1-T)\phi_2=2g_2\in \ker(1+T)$ by (1), (2) implies 
there exists $f_1\in \ker(1-T)$ such that $\phi_1=f_1+g_2$. 
Inserting $\phi_1=f_1+g_2$ and $\phi_2=f_2+g_2$ into the second condition in (\ref{eq:ibuki2}), 
we have $f_1+f_2\neq 0$. 
Then we obtain (\ref{eq:inverse}).
%From (\ref{eq:inverse}), if $U|_\mathcal{L}$ has an eigenvalue $1$, then $T$ has to have an eigenvalue $1$.  

Let us consider about $\ker(1-T)$ in the following. 
Assume $f\in \ker(1-T)$. 
Since $T=K^*SK$, $K^*K=I_V$ and $S^2=I_A$, we have 
	\[ K^*(SKf - Kf)=0,\;\; K^*S(Kf-SK)=0. \]
This implies $SKf-Kf\in \mathcal{L}^\perp\cap \mathcal{L}$. 
Then we have $SKf-Kf=0$. This is equivalent to (\ref{eq:ibuki}), 
which implies $f$ must be a reversible measure. 
Therefore $T$ must be reversible if $\ker(1-T)\neq 0$ by Lemma~\ref{lem:PT}.
Note that the condition $f_1\neq -f_2$  in (\ref{eq:inverse}) implies that either $f_1$ or $f_2$ is not $\bf{0}$. 
This means that if $\ker(1-U|_\mathcal{L})\neq \{\bs{0}\}$, then $\ker(1-T)\neq \{\bs{0}\}$. 
Thus we have shown that the reversibility of $T$ and $m_V\in \ell^2(V)$. 

Now we convert $\varphi$ to $\psi\in \ell^2(A)$ by $L\varphi$ as follows. 
Since $\dim \ker(1-T)=1$, we put $f_2=cf_1$ with some constant $c\neq -1$. 
It is well known that $P$ has a reversible distribution and $G$ is bipartite if and only if $\ker(1+P)\neq 0$. 
In this case, $\dim \ker(1+P)=1$ holds. 
The correspoinding measure $m_V'\in \ker(1+T)$ on a bipartite graph, whose vertices are decomposed into $V=V_0 \sqcup V_1$ such that every edge connects a vertex in $V_0$ to one in $V_1$, 
can be described by $m_V'(u) = (-1)^j m_V(u)$ for any $u\in V_j$ ($j=0,1$). 
We can easily check that $m_V'$ satisfies $m_V'(o(a))\alpha(\bar{a})=-m_V'(t(a))\alpha(a)$. 
%Therefore if $\ker(1+T)\neq 0$, then $\ker(1-T)=\mathbb{C}\{m_V\}$ and $\ker(1+T)=\mathbb{C}\{m_V'\}$. 
Thus we notice that $SKf_1=Kf_1$ because $f_1$ is a reversible measure, and $Kg+SKg=0$ because $g\in\ker(1+T)$ is proportional to $m_V'$. 
Then we have 
	\[ \psi= L\varphi=(Kf_1+SKf_2)+(Kg+SKg)=(1+c)Kf_1,  \]
which implies  $\psi\in \mathbb{C}\{Km_V\}$. 
Then we have reached to the conclusion. 
\end{proof}
%%%%%%%%%%%%%%%%%%%%%%%%%%%%%%%%%%%%%%%%%%%%%%%%%%%%%%%%%%%%%%%%%%%%%%%%
%%%%%%%%%%%%%%%%%%%%%%%%%%%%%%%%%%%%%%%%%%%%%%%%%%%%%%%%%%%%%%%%%%%%%%%%
\section{On the half line}
Let us consider a quantum walk on the half line with the self loop at the origin $j=0$. 
Here the inverse of the self loop $\bar{a}_*$ is regarded as $\bar{a}_*=a_*$ in this paper. 
Recall that $S$ is the flip flop shift. For the self loop $a_*$, we define $S\delta_{a_*}=\delta_{a_*}$. 

It is known that the following quantum coin is a minimal representation for quantum walk on the half line~\cite{CGMV, Oh} which means that
for any time evolution operator $U'$, there exists the minimal representation quantum walk $U$ and diagonal unitary operator $W$ such that 
$U'=WUW^{-1}$ and the necessary number of parameters for the representation is minimum. 
The minimal representation quantum walk's coin assigned at vertex $j\in\{0,1,2,\dots\}$ is 
	\[ C_j=\begin{bmatrix} -\eta_j & \rho_j \\ \rho_j & \bar{\eta}_j \end{bmatrix}, \]
where $\eta_j\in \mathbb{C}$ with $|\eta_j|\leq 1$, and $\rho_j=\sqrt{1-|\eta_j|^2}$. 
For each vertex, we define that ``$|0\rangle:={}^T[1\;\;0]"\in \mathbb{C}^2$ represents the arc from the right direction, 
and for each vertex except the origin, ``$|1\rangle:={}^T[0\;\;1]"\in \mathbb{C}^2$ represents the arc from the left direction, and 
for the origin, $|1\rangle$ represents the self loop. 
Then any quantum walks on the half line can be represented by the sequence of $\eta_j$'s, which are called the Verblunsky parameters. 

The quantum walk treated here is driven by the flip flop shift. 
In the following, we explain this walk is isomorphic to a quantum walk with the moving shift which may be familiar with some researchers on quantum walks. 
Let $(j;R)$ denote the arc whose terminus is $j$ and origin is $j+1$ while $(j+1;L)$ denote the arc whose terminus is $j+1$ and origin is $j$ for every $j\geq 0$. 
Let $(0;L)$ denote the self loop at the origin. 
Consider the moving shift operator $S_m:\ell^2(A)\to \ell^2(A)$ such that 
	\[ (S_m\psi)(x;R)=\psi(x-1;R)\; (x\geq 0);\;\; (S_m\psi)(x;L)=\psi(x+1;L)\; (x\geq 1);   \]
and $(S_m\psi)(0;L)=\psi(0;L)$. 
On the other hand, remark that 
	\[ (S\psi)(x;R)=\psi(x+1;L)\; (x\geq 0);\;\; (S\psi)(x;L)=\psi(x-1;R)\; (x\geq 1); \]
and $(S\psi)(0;L)=\psi(0;L)$. 
Then we have $U=S_m C'$, where 
	\[ C'=S_m^{-1}SC=\oplus \sum_{x\in \mathbb{Z}_+} \begin{bmatrix} \rho_j & \bar{\eta}_j \\ -\eta_j & \rho_j \end{bmatrix}. \]

Now let us proceed to considering what is the underlying of this quantum walk. 
Since the eigenequation for $C_j$ is described by 
	\[ (i\lambda)^2-2\mathrm{Im} (\eta_j) (i\lambda)+1=0, \]
a necessary and sufficient condition for Assumption~1' is that there exists $\kappa\in [-1,1]\subset \mathbb{R}$ which is independent of 
position $j$ such that $\mathrm{Im}(\eta_j)=\kappa$. 
The spectrum of $C_j$ is 
	\[ \mathrm{Spec}(C_j)=\{-i e^{\pm i \phi}\},\] 
where $\cos\phi=\kappa$, and 
$\ker(C_j+ie^{i\phi})=\mathbb{C}[\sqrt{p_j}\;\; \sqrt{q_j}]^\top$. 
Here 
	\begin{align*} 
        p_j = \frac{1}{2}\left(1-\frac{\mathrm{Re}(\eta_j)}{\sqrt{1-\kappa^2}} \right); \;\;
        q_j = \frac{1}{2}\left(1+\frac{\mathrm{Re}(\eta_j)}{\sqrt{1-\kappa^2}} \right).
        \end{align*}
and $\mathrm{Re}(z)$ is the real part of $z\in \mathbb{C}$. 
In Appendix, we show the computational way of this eigenvector. 
Note that this eigenvector corresponds to $\alpha_j$ in the previous section.

From now on, we will consider the following underlying random walk $P=DTD^{-1}$; 
the probabilities associated with the right and left moving at position $j\geq 1$ 
are $p_j$ and $q_j$, respectively, and for $j=0$, $q_0$ is the probability staying at the same position. 
Since the graph is a tree, the reversible measure always exists. Indeed, the reversible measure of $T$ is expressed by
	\begin{equation}\label{eq:revmeasure} 
        \frac{m_V(j)}{m_V(0)}=\sqrt{\frac{p_0\cdots p_{j-1}}{q_1\cdots q_j}}
        \end{equation}
using (\ref{eq:DBC}) recursively. 
%%%%%%%%%%%%%%%%%%%%%%%%%%
\subsection{Review on recurrence properties of random walk}
In this subsection, let us give a short review on recurrence properties of random walks. 
\begin{definition}\label{def:RW}(\cite{Dur,Schi})
Let $T_j$ be a return time to position $j$ of the random walk. 
\begin{enumerate}
\item If $P(T_j<\infty)=1$ and $E(T_j)<\infty$, then we say the random walk is positive recurrent. 
The stationary  state is proportion to $1/E(T_j)$. 
\item If $P(T_j<\infty)=1$ and $E(T_j)=\infty$, then we say the random walk is null recurrent. 
\item If $P(T_j<\infty)<1$, then we say the random walk is transient. 
\end{enumerate}
\end{definition}
Useful necessary and sufficient conditions for recurrence properties of random walks on the half line case are well known as follows: 
\begin{theorem}\label{thm:RW}(\cite{Schi})
Let 
	\[ c_R:=\sum_{j\geq 0} \frac{p_0\cdots p_{j-1}}{q_1\cdots q_j},\;\;c_T:=\sum_{j\geq 1} \frac{q_0\cdots q_{j}}{p_0\cdots p_j}. \]
Then we have 
\begin{enumerate}
\item the random walk is positive recurrent iff $m_V\in \ell^2(\mathbb{Z}_+)$, that is, $c_R<\infty$ (which implies $c_T=\infty$).  
Moreover the stationary distribution at position $j$ is $m_V^2(j)/(1+c_R)$ with $m_V(0)=1$; 
\item the random walk is null recurrent iff $c_R=\infty$ and $c_T=\infty$;  
\item the random walk is transient iff $c_T<\infty$ (which implies $c_R=\infty$). 
\end{enumerate}
\end{theorem}
%%%%%%%%%%%%%%%%%%%%%%%%%%
\subsection{Eigenvalue of QW and positive recurrent random walk}
The statement of Theorem~{\ref{thm:RW}} (1) and (\ref{eq:revmeasure}) imply that 
$T$ has an $\ell^2$-reversible measure if and only if the induced transition matrix $P$ is positive recurrent.
Thus combining it with Lemma~\ref{lem:recQW}, we directly obtain the following proposition.  
Remark that the eigenvalues are rotated by $-e^{i\phi}$ in the setting of Lemma~\ref{lem:recQW}; see Remark~\ref{rem:rotate}. 
\begin{proposition}\label{prop:1}%\cite{Seg}
If $P$ is positive recurrent, then the $(-ie^{i\phi})$-eigenspace of $U$ exists such that $\ker(-ie^{i\phi}-U)=\mathbb{C}\{Km_V\}=\mathbb{C}\{\psi_*\}$. 
Here $\psi_*$ is $\psi_*(0;L)=1$, 
	\[ \psi_*(j;R)=\psi_*(j+1;L)=\sqrt{\frac{p_0\cdots p_j}{q_0\cdots q_j}}\;\mathrm{for}\;j\geq 0. \]
\end{proposition}
Therefore if the initial state has an overlap to this eigenspace and there are no singular continuous spectrum~\cite{SuzSeg}, then the localization happens. 
%%%%%%%%%%%%%%%%%%%%%%%%%%
\subsection{Finite energy flow of QW and transient random walk}
We found that if the underlying random walk is positive recurrent, then $\sigma_p(U)\neq \emptyset$. 
Now how about the opposite case; that is, the underlying random walk is transient ? 
To answer this question, we consider the eigenspace $\mathcal{H}_\pm=\ker(\pm 1-S) \cap \ker K^*$. 
Letting $\psi\in \mathcal{L}^\perp$, from the Kirchhoff condition, we have 
	\begin{equation}\label{eq:patapata} 
        \sqrt{q_j}\psi(j;L)+\sqrt{p_j}\psi(j;R)=0.\;\;(j\geq 0) 
        \end{equation}

First we assume $\psi\in \mathcal{H}_-$. Then $\psi(\bar{a})=-\psi(\bar{a})$ must be hold. 
On the self loop $(0;L)$, since $\overline{(0;L)}=(0;L)$ and the skew symmetry condition, $\psi(0;L)=0$ holds. 
By using (\ref{eq:patapata}) recursively, we have $\psi(a)=0$ for any $a\in A$. 
Thus $\mathcal{H}_-=\boldsymbol{0}$. 

On the other hand, let us consider $\psi\in \mathcal{H}_+$ case. 
By the symmetricity we can put $\psi(0;L):=1$. 
Then (\ref{eq:patapata}) implies $\psi(0;R)=-\sqrt{q_0/p_0}$, and using the symmetricity again we have $\psi(1;L)=-\sqrt{q_0/p_0}$. 
In the same way, (\ref{eq:patapata}) and the symmetricity imply $\psi(1;R)=\psi(2;L)=\sqrt{q_0q_1/p_0p_1}$. 
Taking the same procedure recursively, we obtain
	\begin{align}
        \psi(0;L)=1,\;\; \psi(j;R)=\psi(j+1;L)=(-1)^{j+1}\sqrt{\frac{q_0\cdots q_{j}}{p_0\cdots p_{j}}}\;\; (j\geq 0)
        \end{align}
Therefore $\psi$ is a flow on the half line satisfying the Kirchhoff condition and the symmetricity. 
From the above expression of $\psi$, 
we can notice that the condition of $||\psi||_{\ell^2(A)}<\infty$ is nothing but $c_T<\infty$. 
Then we summarize the statement below. 
\begin{proposition}\label{prop:2}
The underlying birth death process is transient if and only if $\mathcal{H}_-\neq \boldsymbol{0}$. 
The eigenspace is described by $\mathcal{H}_-=\ker(ie^{-i\phi}-U)=\mathbb{C}\xi_*$. 
Here $\xi_*(0;L)=1$, and 
	\[ \xi_*(j;R)=\xi_*(j+1;L)=(-1)^{j+1}\sqrt{\frac{q_0\cdots q_{j}}{p_0\cdots p_{j}}}\;\;(j\geq 0). \]
\end{proposition}
Therefore if the initial state has an overlap to this eigenspace and there are no singular continuous spectrum~\cite{SuzSeg}, 
then the localization happens. 
After all, the appearance of localization of quantum walks is ensured if the underlying random walk is either positive recurrent or transient. 
%%%%%%%%%%%%%%%%%%%%%%%%%%
\section{Application to the simulation of the dynamics on topological insulator}

%%%%%%%%%%%%%%%%%%%%%%%%%%%%%%%%%%%%%%%%%%%%%%%%%%%%%%%%%%%%%%%%%%%%%%%%
%%%%%%%%%%%%%%%%%%%%%%%%%%%%%%%%%%%%%%%%%%%%%%%%%%%%%%%%%%%%%%%%%%%%%%%%
\subsection{AE Model}
We consider the infinite graph cutting $2$-dimensional lattice at the $x=0$ line and add the self loops to every vertex located in the $x=0$ edge. 
Then the vertex set of this graph is $\mathbb{Z}_+\times \mathbb{Z}$. 
We denote $\partial V=\{(0,y) \;|\; y\in \mathbb{Z}\}\in V$. 
The total Hilbert space is denoted by $\ell^2(A)\cong \ell^2(V;\mathbb{C}^2)$. 
A quantum walker moves vertical and horizontal directions of $G$ alternatively. 
The quantum coins for vertical and horizontal directions are parametrized by $\alpha, \beta\in \mathbb{R}$, respectively. 
The time evolution starting from the self loop at the origin 
is described by the iteration of the unitary map $\Gamma: \ell^2(V;\mathbb{C}^2)\to \ell^2(V;\mathbb{C}^2)$~\cite{AE}:
        \begin{multline}
        (\Gamma \varphi)(x,y)
        	= Q_\alpha Q_\beta\varphi(x-1,y-1)+Q_\alpha P_\beta\varphi(x-1,y+1) \\ +P_\alpha Q_\beta\varphi(x+1,y-1)+P_\alpha P_\beta\varphi(x+1,y+1),\;\;((x,y)\notin \partial V),
\label{eq:Gamma_D}
        \end{multline}
and
        \begin{multline}
        (\Gamma\varphi)(x,y)
        	= S_\alpha Q_\beta\varphi(x,y-1)+S_\alpha P_\beta\varphi(x,y+1) \\ +P_\alpha Q_\beta\varphi(x+1,y-1)+P_\alpha P_\beta\varphi(x+1,y+1),\;\;((x,y)\in \partial V).
\label{eq:Gamma_D2}
        \end{multline}
Here for two-dimensional $\gamma$-rotation matrix 
	\[ H_\gamma=\begin{bmatrix}  \cos \gamma & -\sin \gamma \\ \sin\gamma & \cos \gamma \end{bmatrix} \;\;(\gamma\in[0,2\pi)), \]
we define 
	\[ P_\gamma=|0\rangle\langle 0|H_\gamma,\; Q_\gamma=|1\rangle\langle 1|H_\gamma,\; S_\gamma=|1\rangle\langle 0|H_\gamma. \]

%%%%%%%%%%%%%%%%%%%%%%%%%%
\subsection{Review on the edge state computed by the spectral analysis on CMV matrix}
By the translation invariant of this model with respect to the vertical direction, we take the Fourier transform with respect to $y$; 
	\[ \hat{\varphi}(x;k):=(\mathcal{F}\varphi)(x;k)=\sum_{y\in \mathbb{Z}} \varphi(x,y)e^{iky}. \]
The dynamics is the collection of some quantum walks on the half line with respect to the wave number $k\in[0,2\pi)$ discussed in the previous section. 
Moreover the unitary map on this infinite quotient graph; the half line, is isomorphic to the CMV matrix as follows.  
\begin{proposition}(\cite{EKOS})
Let $\mathcal{C}_k$ be the CMV matrix whose Verblunsky parameters are $(\eta(k),0,\eta(k),0,\dots)$, where $\eta(k)=-\sin(\alpha+\beta)\cos k+i\sin(\alpha-\beta)\sin k$. 
\[ (\Gamma^n\varphi)(x,y)=\int_{0}^{2\pi} \left( \Lambda_k^{-1}(\mathcal{C}_k^\top)^n\Lambda_k \hat{\varphi} \right)(x)e^{-iky}\frac{dk}{2\pi},  \]
where $\Lambda_k^{-1}: \ell^2(\mathbb{Z}_+)\to \ell^2(\mathbb{Z}_+;\mathbb{C}^2)$ is defined by 
\[ (\Lambda_k^{-1}f)(j)=\begin{bmatrix} e^{-i\omega(2j+1)}f(2j+1) \\ e^{-i\omega(2j)}f(2j) \end{bmatrix}. \]
Here \[ \omega(2j)=-j\;\mathrm{arg}(\langle 0| \hat{H}_k |0 \rangle),\;\;\omega(2j+1)=(j+1)\;\mathrm{arg}(\langle 1| \hat{H}_k |1 \rangle), \]
with 
\[ \hat{H}_k=\begin{bmatrix} 
        e^{-ik}\cos\alpha\cos\beta-e^{ik}\sin\alpha\sin\beta & -e^{-ik}\cos\alpha\sin\beta-e^{ik}\sin\alpha\cos\beta \\ 
        e^{-ik}\sin\alpha\cos\beta+e^{ik}\cos\alpha\sin\beta & -e^{-ik}\sin\alpha\sin\beta+e^{ik}\cos\alpha\cos\beta
        \end{bmatrix}.  \] 
\end{proposition}
The problem is essentially reduced to the spectral analysis on the CMV matrix $\mathcal{C}_k$. 
By the standard spectral analysis on the CMV matrix, we obtain the following spectral decomposition of $\sigma(\mathcal{C}_k)$. 
\begin{lemma}(\cite{EKOS})
For fixed $k\in[0,2\pi)$, the spectrum of $\mathcal{C}_k$ is decomposed into continuous spectrum $\sigma_c^{(k)}$ 
	and point spectrum $\sigma_p^{(k)}$, that is, $\sigma(\mathcal{C}_k)=\sigma_c^{(k)} \sqcup \sigma_p^{(k)}$, where
		\begin{align} 
        	\sigma_c^{(k)} &= \{ e^{i\theta} \;|\; |\cos\theta|\leq \rho(k)\} \\
        	\sigma_p^{(k)} &= 
                \begin{cases} 
                \{e^{i \theta_0}(k)\} & \text{: $\sin(\alpha-\beta)\neq 0$, $k\notin \{\pi/2,3\pi/2\}$, }  \\ 
                \emptyset & \text{: otherwise.} 
                \end{cases}
        	\end{align}
	Moreover we have 
        	\begin{equation}
                \sigma(\Gamma)=\bigcup_{k\in[0,2\pi)}\left( \sigma_c^{(k)} \sqcup \sigma_p^{(k)} \right). 
                \end{equation}
Here 
	\begin{align}
        \rho(k) & := \sqrt{1-|\eta(k)|^2}, \label{rh}\\
        m_0(k) & := m_0
        =\begin{cases}
        |\mathrm{R}(\eta(k))|/\sqrt{1-\mathrm{Im} ^2(\eta(k))} & \text{: $\rho(k)\neq 0$} \\
        1 & \text{: $\rho(k)=1$,}
        \end{cases} \label{mass} \\
        \theta_0(k) & :=
        \begin{cases}
        \arcsin (-\mathrm{Im} (\eta(k))) & \text{: $\mathrm{R}(\eta(k))\geq 0$,} \\ 
        \pi-\arcsin(-\mathrm{Im}(\eta(k))) & \text{: $\mathrm{R}(\eta(k))< 0$,}
        \end{cases} \label{theta_0}
        \end{align}
\end{lemma}
Defining $\theta_c(k)$ as $\arccos(\rho(k))$, that is,  
		\begin{equation}\label{acpart} 
                \theta_c(k)=\arccos \left(\sqrt{\cos^2(\alpha-\beta)-\sin 2\alpha\sin 2\beta \cos^2k}\right), 
                \end{equation}
we find that the continuous spectrum of $\mathcal{C}_k$ is $\{e^{i\theta} \;|\; \theta\in [\theta_c(k),\pi-\theta_c(k)]\cup [\pi+\theta_c(k),2\pi-\theta_c(k)]\}$, 
and the point spectrum is $\{\theta_0(k)\}$ if it exists. 
Then the the dispersion relations between the wave number $k$ vs quasi-energy $\theta$ is described as follows. 
\begin{proposition}(\cite{EKOS})
The dispersion relation for bulk and edge states, $Bu$ and $Ed$, are
	\begin{align} 
        Bu &= \bigcup_{k\in[0,2\pi)} [\theta_c(k),\pi-\theta_c(k)] \cup [\pi+\theta_c(k),2\pi-\theta_c(k)] \label{bulk} \\
        Ed &= \begin{cases}
                \{ (k,\theta_0(k)) \;|\; k\in[0,2\pi)\setminus \{\pi/2,3\pi/2\} \} & \text{: $\sin(\alpha-\beta)\neq 0$, } \\ 
                \emptyset & \text{: $\sin(\alpha-\beta)=0$, }
              \end{cases}  \label{edge}
        \end{align}
\end{proposition}
Remark that $\theta_0(k)$ is a monotone function and having two jumps at $k=\pi/2$ and $k=3\pi/2$. 
See Fig.~\ref{fig:dispersion} for $(\alpha, \beta)=(5\pi/4,\pi/6)$ case. 

%%%%%%%%%%%%%%%%%%%%%%%%%%
\subsection{Simple proof of the edge state using the recurrence properties of the underlying random walk}
The reduced quantum walk in the Fourier space for fixed $k$ is the quantum walk with the self loop at the origin whose quantum coin assigned at $j$ is 
	\[ C_j=\begin{bmatrix} -\eta(k) & \rho(k) \\ \rho(k) & \bar{\eta}(k) \end{bmatrix}. \]
Then we have 
	\[ p:=p_j=\frac{1}{2}\left(1-\frac{\mathrm{Re}(\eta(k))}{\sqrt{1-\kappa^2}}\right),\;q:=q_j=\frac{1}{2}\left(1+\frac{\mathrm{R}(\eta(k))}{\sqrt{1-\kappa^2}}\right). \]
Let $P_k$ be the transition operator of underlying random walk on the half line parameterized by $k$. 
From \cite{Iga-Oba}, the number of eigenvalues of $P_k$ is at most $1$. %, that is, $|\sigma_p(P_k)|\leq 1$. 
We can easily to see that
	\begin{enumerate}
        \item $\mathrm{Re}(\eta(k))>0$ iff $P_k$ is positive recurrent; 
        \item $\mathrm{Re}(\eta(k))=0$ iff $P_k$ is null recurrent; 
        \item $\mathrm{Re}(\eta(k))<0$ off $P_k$ is transient.
        \end{enumerate}
Then just checking the recurrent properties of the underlying random walk, 
we can compute the edge state without a direct spectral analysis on the CMV matrix, and 
we can also obtain the eigenvectors as follows: 
%%%%%%%%%%%%%%%%%%
\begin{theorem}\label{thm:dispersion}
Put $\kappa=\mathrm{Im} (\eta(k))$. Let $\sigma_p(\mathcal{C}_k)$ be the set of point spectrum of the CMV matrix $\mathcal{C}_k$. 
\begin{enumerate}
\item If $P_k$ is positive recurrent which is equivalent to $\mathrm{Re}(\eta(k))>0$, 
then $\sigma_p(\mathcal{C}_k)=\{e^{i\arccos(-\kappa)}\}$, $\ker(e^{i\arccos(-\kappa)}-\mathcal{C}_k)=\mathbb{C}\psi_*$ and $m=1/||\psi_*||^2=1/(q-p)$. 
\item If $P_k$ is transient which is equivalent to $\mathrm{Re}(\eta(k))<0$, 
then $\sigma_p(\mathcal{C}_k)=\{e^{i(\pi-\arccos(-\kappa))}\}$, $\ker(e^{i(\pi-\arccos(-\kappa))}-\mathcal{C}_k)=\mathbb{C}\xi_*$ and $m=1/||\xi_*||^2=1/(p-q)$.
\item If $P_k$ is null recurrent which is equivalent to $\mathrm{Re}(\eta(k))=0$, 
then $\sigma_p(\mathcal{C}_k)=\emptyset$. 
\end{enumerate} 
\end{theorem}
\begin{proof}
Remark that the number of eigenvalues of $P_k$ is at most $1$, and if the eigenvalue exists, then the underlying random walk is positive recurrent. 
For the cases (1) and (2), applying Propositions~\ref{prop:1} and \ref{prop:2} for $\eta_j=\eta(k)$ case, we obtain the desired conclusion.  
For the case (3), since $C_j=e^{-i\kappa}I_2$, the walk becomes a free walk. Then the point spectrum does not exist. 
\end{proof}
It is easily computed that $\mathrm{Re}(\eta(k))>0$ if and only if $k\in(0,\pi/2)\cup (3\pi/2,2\pi)$ and 
$\mathrm{Re}(\eta(k))<0$ if and only if $k\in(\pi/2,3\pi/2)$ and $\mathrm{Re}(\eta(k))=0$ if and only if $k\in\{0,\pi/2\}$ in the region of $k\in[0,2\pi)$. 
See Fig.~\ref{fig:dispersion} for $(\alpha,\beta)=(5\pi/4,\pi/6)$ case, the red curve which has jumps at $k=\pi/2$ and $k=2\pi/3$ 
depicts $\sigma_p(\mathcal{C}_k)$ as a function of $k\in[0,2\pi)$ following Theorem~\ref{thm:dispersion}. 
\begin{figure}[!ht]
\begin{center}
	\includegraphics[width=100mm]{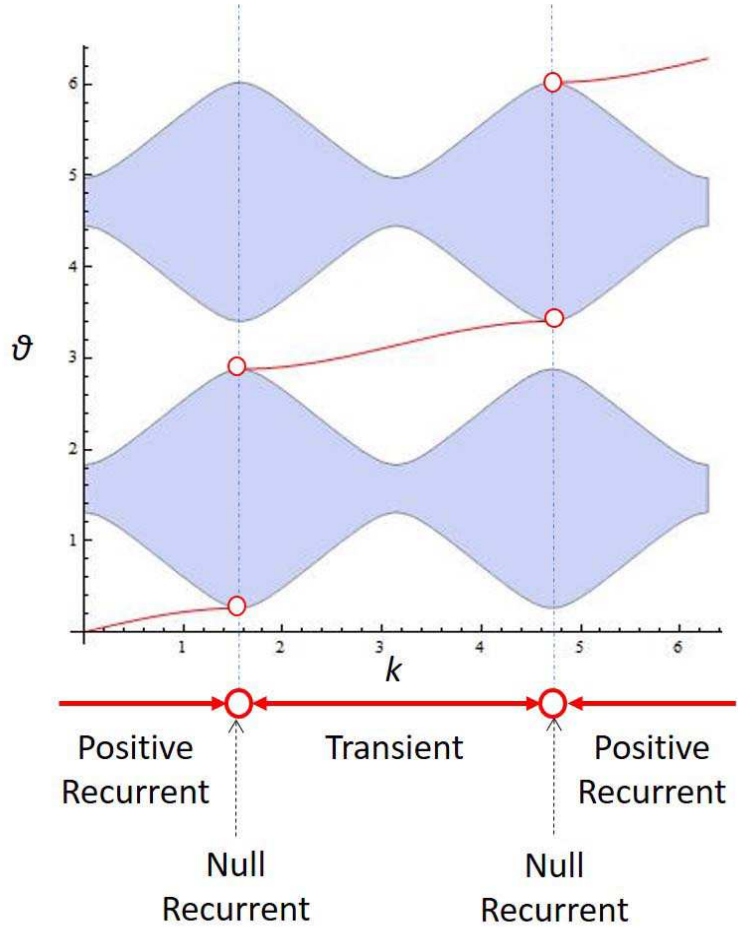}
\end{center}
\caption{ The dispersion relation between the wave number $k$ and quasi-energy $\theta$. 
The red curve depicts the edge state and the blue region depicts the bulk state. 
The region $k\in [0,\pi/2)\cup (3\pi/2, 2\pi]$ corresponds to the positive recurrent part of the underlying random walk while
the region $k\in (\pi/2,3\pi/2)$ corresponds to the transient part, and the boundaries; that is, $k\in\{\pi/2,3\pi/2\}$ correspond to the null recurrent part. 
}
\label{fig:dispersion}
\end{figure}

%%%%%%%%%%%%%%%%%%%%%%%%%%%%%%%%%%%%%%%%%%%%%%%%%%%%%%%%%%%%%%%%%%%%%%%%
%%%%%%%%%%%%%%%%%%%%%%%%%%%%%%%%%%%%%%%%%%%%%%%%%%%%%%%%%%%%%%%%%%%%%%%%
\noindent \\
\noindent {\bf Acknowledgements} \\
YI was supported by the Grant-in-Aid for Young Scientists (B) of Japan Society for the Promotion of Science (Grant No.~16K17652).
ES acknowledges financial support from the Grant-in-Aid of Scientific Research (C) Japan Society for the Promotion of Science (Grant No.~19K03616).

%%%%%%%%%%%%%%%%%%%%%%%%%%%%%%%%%%%%%%%%%%%%%%%%%%%%%%%%%%%

\appendix
\def\thesection{Appendix \Alph{section}}
\renewcommand{\theequation}{A.\arabic{equation}}
\setcounter{equation}{0}

% Appendix 

\section{Computation of $\ker(C_j+ie^{i\phi})$}\label{pfalternative}
The $(-ie^{i\phi})$-eigenvector of $C_j$ is expressed by 
\[ \ker(C_j+ie^{i\phi})=\mathbb{C}\begin{bmatrix} \rho_j \\ \eta_j-ie^{i\phi} \end{bmatrix}. \]
Recall that $\cos\phi=\kappa=\mathrm{Im}(\eta_j)$ and $e^{i\phi}=\kappa+i\sqrt{1-\kappa^2}$. 
Then each element of the eigenvector can be deformed by 
	\[ \rho_j=\sqrt{ (1-\kappa^2)-\mathrm{Re}(\eta_j)^2 }=A_+A_-, \; \eta_j-ie^{i\phi}
        	 =\sqrt{ 1-\kappa^2 }+\mathrm{Re}(\eta_j)  =A_+^2, \]
respectively. 
Here we put 
	\[ A_\pm :=\sqrt{\sqrt{1-\kappa^2}\pm \mathrm{Re}(\eta_j)}. \]
Therefore the normalized eigenvector is expressed by 
	\[ \frac{1}{\sqrt{A_+^2+A_-^2}}\; {}^T[A_-\; A_+]={}^T[\sqrt{p_j}\; \sqrt{q_j}], \] where  
	\[ p_j=\frac{1}{2}\left( 1-\frac{\mathrm{Re}(\eta_j)}{\sqrt{1-\kappa^2}} \right),\;q_j=\frac{1}{2}\left( 1+\frac{\mathrm{Re}(\eta_j)}{\sqrt{1-\kappa^2}} \right). \]
\end{document}